\documentclass{article}
\usepackage[utf8]{inputenc}
\usepackage{fullpage}
\usepackage{amsmath}
\usepackage{amssymb}
\usepackage{amsthm}
\usepackage{dsfont}
\usepackage{physics}
\usepackage{tikz}
\usepackage{array}
\usepackage{comment}
\usepackage{float}
\usepackage{authblk}

\usepackage[colorlinks=true,linkcolor=blue,allcolors=blue]{hyperref}
\usepackage[capitalize,nameinlink]{cleveref}
\usepackage{todonotes}
\usepackage{lmodern}
\usetikzlibrary{decorations.pathmorphing}

\tikzset{snake it/.style={decorate, decoration=snake}}

\newtheorem{theorem}{Theorem}[section]
\newtheorem{definition}[theorem]{Definition}

\newtheorem{remark}[theorem]{Remark}
\newtheorem{lemma}[theorem]{Lemma}
\newtheorem{cor}[theorem]{Corollary}

\newtheorem{fact}[theorem]{Fact}

\newtheoremstyle{named}{}{}{\itshape}{}{\bfseries}{.}{.5em}{\thmnote{#3}}
\theoremstyle{named}

\newcommand{\poly}{\operatorname{poly}}

\newcommand{\supp}{\operatorname{supp}}

\newcommand{\mcP}{\mathcal{P}}
\newcommand{\mcC}{\mathcal{C}}

\newcommand{\mcB}{\mathcal{B}}

\newcommand{\mcQ}{\mathcal{Q}}

\newcommand{\mcE}{\mathcal{E}}

\newcommand{\mcF}{\mathcal{F}}

\newcommand{\mfe}{\mathrm{MF}}
\newcommand{\rmfe}{\mathrm{RMF}}
\newcommand{\mbF}{\operatorname{\mathbb{F}}}

\newcommand{\bs}[1]{\boldsymbol{#1}}
\newcommand{\rowspan}{\operatorname{rowsp}}

\title{Good binary quantum codes with transversal CCZ gate}
\author{Quynh T. Nguyen\thanks{Harvard University. \href{mailto:qnguyen@g.harvard.edu}{qnguyen@g.harvard.edu}}}

\begin{document}

\maketitle

\begin{abstract}
    We give an asymptotically good family of quantum CSS codes on qubits with a transversal $CCZ$ gate, meaning that the parallel logical $CCZ$ on all logical qubits is performed by parallel physical $CCZ$s on (a subset of) physical qubits. The construction is based on the observation that any classical code satisfying a multiplication property can be used to construct a quantum CSS code with transversal (qudit) $CCZ$. To obtain a constant-rate and linear-distance family, we then instantiate this construction with a classical good family of algebraic-geometry codes on a non-binary, but constant-sized, alphabet. Finally, we use a technique from the arithmetic secret sharing literature to reduce the alphabet to binary. As a corollary, the constructed code family provides a magic state distillation scheme with constant space overhead.
\end{abstract}

\section{Introduction}
Transversal gates are the simplest way to fault-tolerantly perform encoded operations on a quantum code. While there are many quantum code families supporting transversal Clifford gates, constructing a family supporting non-Clifford gates such as $T$ or $CCZ$ appears to be much more difficult, even for non-LDPC codes. To our knowledge, there have been no known constant-rate quantum code families with a transversal non-Clifford gate and having growing distance~\cite{bravyi2005universal, zeng2007transversality, bravyi2012magic, campbell2012magic, campbell2014enhanced, haah2018towers, vasmer2022morphing, haah2017magic, hastings2018distillation, haah2018codes, anwar2012qutrit, vuillot2022quantum}. Such codes would be useful for magic state distillation (MSD), and constructing them would perhaps be a good step towards finding good quantum LDPC code with transversal non-Clifford~\cite{zhu2023non}.

In a prior work, Krishna and Tillich~\cite{krishna2019towards} constructed a constant-rate and linear distance qudit code family that supports a transversal qudit non-Clifford gate; however, their code has the qudit dimension growing with the code size. For the qubit case, the previously best asymptotic result was due to Hastings and Haah~\cite{hastings2018distillation}, who constructed a vanishing rate and sublinear distance code family, that implies a state-of-art MSD overhead of $O(\log^{0.68}(1/\varepsilon))$. Most attempts to construct quantum codes with non-Clifford gate are based on classical algebraic code constructions, and are often formulated in the triorthogonality framework. More generally, these classical codes have the so-called multiplication property, meaning that the component-wise product of two codewords are codewords of some relevant code. For example, in the Reed-Solomon (RS) code, multiplying two codewords is the same as multiplying two degree-$d$ polynomials, yielding a degree-$2d$ polynomial, which is a codeword in a larger RS code. Furthermore, it is well-known that the dual of an RS code is an RS code with known degree. These two facts were used in~\cite{krishna2019towards} in a clever way to obtain the quantum code mentioned above. The issue with RS codes is that they have growing alphabet. And thus it is natural to ask if there is already a solution to this in the classical coding literature. Goppa's algebraic-geometry (AG) codes~\cite{goppa1981codes} are indeed known to address this issue, leading to improvements in a wide range of applications in cryptography, algebraic complexity theory, distributed computation etc. wherever a fixed-size alphabet is preferred~\cite{couvreur2021algebraic}. AG codes generalize both Reed-Solomon and Reed-Muller codes. Furthermore, they still have the multiplication property: the product of two AG code is included in, and in many situations is equal to, an AG code. They behave well under duality: the dual of an AG code is an AG code. Given this ubiquity of classical AG codes, quantum code constructions based on AG codes have of course been defined and studied before~\cite{chen2001some,  chen2001asymptotically, bartoli2021certain, hernando2020quantum, li2009family, la2017good, pereira2021entanglement}, in fact giving the first explicit families of good quantum codes~\cite{ashikhmin2001asymptotically}.
However, previous works did not study logical gates on these codes.

In this paper, we revisit the technique by Krishna and Tillich~\cite{krishna2019towards} of constructing a quantum code by puncturing a classical code with multiplication property and apply it to an asymptotically good family of AG codes with fixed-size alphabet. In~\Cref{sec:gate}, we give a generalization of their construction to the case of power-of-2 qudits and multi-qudit logical gates, and observe that the number of multiplications supported by the underlying classical code essentially transfers into the transversality of qudit phase gates with a corresponding degree (e.g., $CCZ$ has degree 3). In particular, we show that the constructed quantum code supports transversal implementations of the (qudit) $CCZ$ gate, and other gates that we call low-degree phase gates. In~\Cref{sec:AG}, we then instantiate this construction by modifying the code parameters of an asymptotically good classical AG code family from the book of Stichtenoth~\cite{stichtenoth2009algebraic}, giving an asymptotically good qudit quantum code family over a fixed-size power-of-2 alphabet with the mentioned transversality properties. Finally, in~\Cref{sec:qubit}, we use techniques from the area of arithmetic secret sharing scheme (in which AG codes are also ubiquitous)~\cite{cascudo2009asymptotically, cascudo2018amortized}, called multiplication-friendly embeddings, to convert the constructed qudit code into a qubit code while keeping the family asymptotically good, and furthermore, the logical $CCZ$ gate on the resulting qubit code remains transversal.

\paragraph{Application.} The constructed code family can be used in an MSD scheme with constant space overhead as follows. To distill the magic states $\ket{CCZ}=CCZ\ket{+}^{\otimes 3}$ to a target error rate $\varepsilon$, we take a code instance of parameters $[[N,K=\Theta(N),D=\Theta(N)]]$ with $N=\Theta(\log 1/\varepsilon)$ from the code family. We then prepared three code blocks, each encoding the state $\ket{+}^{\otimes K}$. Then using $\Theta(N)$ noisy $\ket{CCZ}$ states with constant noise rate (sufficiently below the code's relative distance) we transversally implement the parallel logical $CCZ$ to obtain $\ket{\overline{CCZ}}^{\otimes K}$. Performing error correction and unencoding the three code blocks, we obtain $K=\Theta(N)$ good $\ket{CCZ}$ states with error rate suppressed to $e^{-\Theta(D)}=\varepsilon$ as desired. We note that although the space overhead here is constant, we have not studied the explicit encoding circuit and efficient decoder for the code and thus the time complexity is potentially $\exp(N)=\poly(1/\varepsilon)$ for the error correction step.


\paragraph{Open questions.} We list here several open questions. First, how to perform a universal gate set fault-tolerantly on our code? Fault-tolerant universal logical implementations were provided in~\cite{aharonov1999fault} on a family of quantum Reed-Solomon codes, where the authors devised a quantum degree reduction gadget to implement non-transversal gates like Hadamard. Is there
an analogue of their degree reduction gadget for quantum AG codes? Second, can we get a qubit code directly, instead of first getting a qudit code and then performing alphabet reduction to qubits? Third, it is interesting to study applications of quantum AG codes and multiplication property in quantum cryptography, as their classical analogues  have many such applications~\cite{couvreur2021algebraic}.

\paragraph{Concurrent works.} Independently, Ref.~\cite{wills2024constantoverheadmagicstatedistillation} construct an explicit constant-overhead MSD scheme based on AG codes. Our work differs from theirs in that we obtain transversal qudit $CCZ$ gate (rather than single-qudit gate like they do), and we eventually get a qubit code with strongly transversal qubit $CCZ$, answering an open question in~\cite{wills2024constantoverheadmagicstatedistillation}. We did not explicitly describe an MSD scheme with efficient decoder like~\cite{wills2024constantoverheadmagicstatedistillation}, and their scheme also works for the weak transversality case. Similar results were also obtained concurrently by Ref.~\cite{golowich2024asymptoticallygoodquantumcodes}. Very similarly to us,~\cite{golowich2024asymptoticallygoodquantumcodes} obtain a transversal qubit $CCZ$ gate using AG codes and multiplication-friendly embeddings. They additionally apply their construction to RS codes.
Finally, Christopher A. Pattison and the author~\cite{nguyen2024} concurrently construct an MSD protocol not based on AG codes, that gives a $\log^{\gamma}(1/\varepsilon)$ spacetime overhead (rather than just space), where $\gamma \rightarrow 0$ and further apply it to construct a quantum fault tolerance scheme with low spacetime overhead.

\section{Transversal low-degree phase gates from multiplication property}\label{sec:gate}
In this section, we present a construction of a good qudit code that supports transversally what we call `low-degree' phase gates, such as $CCZ$. The construction is based on a classical code with a multiplication property. We first review some facts about finite fields, prime-power qudits, and qudit CSS codes. For more details, we refer to~\cite{aharonov1999fault, cross2008, gottesman2024}.

In this work, we consider linear classical codes, often denoted $\mcC$, and quantum CSS codes, denoted $\mcQ$, over a finite field $\mathbb{F}_q$, where $q$ is a power of 2 (although we believe this paper can be generalized to arbitrary characteristic).
\paragraph{Prime-power qudits} Consider the field $\mbF_{q}$ where $q=p^m$. Let $\omega$ be the $p$-th root of unity. For $a, b, x \in \mbF_q$, the qudit Pauli operators are defined as
\begin{align*}
    X^{(q)}_a\ket{x} = \ket{x+a}, \qquad 
    Z^{(q)}_b\ket{x} = \omega^{\tr(bx)}\ket{x},
\end{align*}
where $\tr(x) = 1 +x^2+\hdots +x^{p^{m-1}}$ is the trace map $\mbF_{p^m} \mapsto \mbF_p$, which is well-known to be $\mbF_p$-linear~\cite{mullen2013handbook}. Here, a map $f: \mbF_q \mapsto \mbF_p$ is said to be $\mathbb{F}_p$-linear if $f(ax+by)=af(x)+bf(y)$ when $a,b \in \mathbb{F}_p$. We have the commutation relation
\begin{align*}
    X^{(q)}_a Z^{(q)}_b = \omega^{\tr(ab)} Z^{(q)}_b X^{(q)}_a.
\end{align*}

In this paper, we use the superscript $(q)$ to distinguish the qudit and qubit versions of a gate. When $q=2$ we omit this notation, so $X$, $Z$, $CCZ$, etc., denote the usual qubit operators.

\paragraph{Non-Pauli qudit gates} The generalizations of $CNOT, H, CCZ,$ and Toffoli for qudits can be defined. For $p >2$, there may be more than one generalized versions~\cite{aharonov1999fault}, but from now on we only consider $\operatorname{char} \mathbb{F}_q=p=2$. In this case, we have
\begin{align*}
    CNOT^{(q)} \ket{x}\ket{y} = \ket{x}\ket{x+y}, \qquad H^{(q)}\ket{x} = \frac{1}{\sqrt{q}} \sum_{z \in \mbF_q} (-1)^{\tr(zx)} \ket{z},
\end{align*}
\begin{align*}
    CCZ^{(q)}\ket{x}\ket{y}\ket{z}= (-1)^{\tr(xyz)} \ket{x}\ket{y}\ket{z}, \qquad CCNOT^{(q)} \ket{x}\ket{y}\ket{z} = \ket{x}\ket{y}\ket{z+xy}.
\end{align*}
In this work, we also define gates which we call \emph{low-degree phase gates}. Consider the degree 3 case first, these are diagonal gates of the form
\begin{align}
    U^{(q)}_{f,g} \ket{x}\ket{y}\ket{z} = (-1)^{f(g(x,y,z))}\ket{x}\ket{y}\ket{z},
    \label{eq:low-degree-gate}
\end{align}
where $g: \mbF_q^3 \mapsto \mbF_q$ is a function involving monomials in $x,y,z$ of degree at most $3$, and $f: \mbF_q \mapsto \mbF_2$ is an $\mbF_2$-linear function. For example, $CCZ^{(q)}$ is the degree-3 gate with $g(x,y,z)=xyz$ and $f(x)= \tr(x)$. Generalizations to higher degrees $t> 3$ are the natural ones, where $g$ has degree $t$.

Given two linear codes $C_1, C_2 \in \mbF_q^n$ such that $C_2^\perp \subseteq C_1$, the quantum CSS code CSS$(C_1,C_2)$ is the span of the states of the form
\begin{align}
    \ket{c+C_2^\perp} \triangleq \frac{1}{\sqrt{|C_2^\perp|}} \sum_{c' \in C_2^\perp} \ket{c + c'}, \qquad c \in C_1.
\end{align}

Let us now describe the notion of \emph{multiplication property} from the classical coding literature. We will show in~\Cref{lem:ccz} that classical codes with this property can be used to construct quantum CSS codes supporting transversal implementations of low-degree phase gates. This idea was introduced in~\cite{krishna2019towards}, but with a different name and restricted to prime fields and single-qudit gates (corresponding to single-variable polynomials in~\Cref{eq:low-degree-gate}). We follow~\cite{couvreur2021algebraic} to go with the terminology `multiplication property', which is often used in their complexity theory and cryptography applications.

\begin{definition}[Star product] The square (star product) of a linear code over the finite field $\mathbb{F}_q$ is defined to be $\mcC^{\star 2} = \operatorname{span}_{\mathbb{F}_q} \{ c \star c': c, c' \in \mcC \}$, where $\star$ denotes componentwise multiplication over $\mathbb{F}_q$. The $t$-star product for $t\geq 3$ is defined analogously.
\end{definition}

\begin{definition}[Multiplication property, see, e.g.,~\cite{couvreur2021algebraic}]\label{def:mult-prop}
A linear code is said to satisfy the multiplication property if $\mcC^{\star 2}$ is contained in another relevant code. In particular, in this paper, we take this definition to be $\mcC^{\star 2} \subseteq \mcC^\perp$. Equivalently, this means that for any three codewords $ x,y,z \in \mcC$, we have $|x \star y \star z| \triangleq \sum_{i=1}^n x_iy_iz_i=0$ over $\mathbb{F}_q$. The $t$-multiplication property is defined analogously, but in this paper by `multiplication property' we mean the case $t=2$.
\end{definition}

The following lemma shows that a classical code with multiplication property can be used to construct a quantum code with certain transversal gates.

\begin{lemma}[Qudit code with transversal low-degree phase gates]
\label{lem:ccz}
Given a classical code $\mcC$ $[n,k,d]_{\mathbb{F}_q}$ with the multiplication property, and additionally, containing the all-1's word, we can construct a quantum code $\mcQ$ $[[N,K,D]]_{\mathbb{F}_q}$ such that the transversal qudit $(CCZ^{(q)})^{\otimes N}$ implements $(\overline{CCZ^{(q)}})^{\otimes K}$. More generally any logical degree-3 phase gates as defined in~\Cref{eq:low-degree-gate} can also be implemented transversally in a similar way. Furthermore, the quantum code parameters are $N=n-K$ and $D_X \geq d-K$ and $D_Z \geq  \operatorname{dist}(\mcC'^\perp)$, where $\mcC'$ is a `$K$-shortened' version\footnote{A punctured version $\mcC_1$ of a code $\mcC_0$ is a code obtained from $\mcC_0$ by removing some coordinates. A shortened code $\mcC_2$ of $\mcC_0$ is obtained from $\mcC_1$ by restricting to codewords whose removed coordinates are initially zero. See~\Cref{eq:shortened-code} for an example.} of $\mcC$ defined in~\Cref{eq:shortened-code}. Here $K$ can in principle take any value $\leq k$, but $\operatorname{dist}(\mcC'^\perp)$ will depend on $K$.
\end{lemma}
\begin{proof}
    We follow the construction of Krishna and Tillich \cite{krishna2019towards}. 
    First, we express the generating matrix $G \in \mathbb{F}_q^{k \times n}$ of $\mcC$ in the following form
    \begin{align}
        G = \begin{pmatrix}
            \mathbf{1}_K & H_1\\
            0 & H_0
        \end{pmatrix},
        \label{eq:puncture}
    \end{align}
    where $K \leq k$, $H_1 \in \mathbb{F}_q^{K \times (n-K)}$, and $H_0 \in \mathbb{F}_q^{(k-K) \times (n-K)}$. Then, for any three rows $r_a,r_b,r_c$ indexed by $a, b, c \in [k]$ of $\begin{pmatrix}
        H_1 \\
        H_0
    \end{pmatrix}$ (below we slightly abuse notation and denote $a \in H_1$ to mean $a \in [K]$, and $a\in H_0$ to mean $K+1 \leq a \leq k$), the following holds
    \begin{align}
        |r_a \star r_b \star r_c| = \begin{cases}1 & \text{ if }  a=b=c \in H_1\\
        0 & \text{ otherwise}
        \end{cases}.
        \label{eq:punctured-star}
    \end{align}
    Indeed, due to the multiplication property, any three rows $R_a, R_b, R_c$ of $G$ satisfy $|R_a \star R_b \star R_c| = 0$. Let $r_a, r_b, r_c$ be the restrictions of $R_a, R_b, R_c$ onto the last $n-K$ columns. It can be readily verified that $|R_a \star R_b \star R_c| = |r_a \star r_b \star r_c|$ unless $a= b=c \in H_1$. When that is the case, we have $|R_a \star R_b \star R_c| = 1+ |r_a \star r_b \star r_c|$, implying $|r_a \star r_b \star r_c|=1$.
    
    Moreover, applying the same argument but taking $R_c$ to be the all-1's word, we obtain that
    \begin{align}
        |r_a \star r_b| = \begin{cases}
        1 & \text{ if }  a=b \in H_1\\
        0 & \text{ otherwise}
    \end{cases}.
    \label{eq:inner-prod}
    \end{align}
    Similarly, replacing both $R_b$ and $R_c$ by the all-1's word, we have
    \begin{align}
    |r_a| = \begin{cases}
        1 & \text{ if }  a \in H_1\\
        0 & \text{ otherwise}
    \end{cases}.
    \label{eq:weight}
    \end{align}
    
    It follows from the above that the rows of $H_0, H_1$ are linearly independent and $\rowspan(H_1) \cap \rowspan(H_0) = \{0\}$. Let us show this. Indeed, $H_0$'s rows are linearly independent by definition of $G$. For $H_1$, suppose $x=\sum_{a \in H_1} x_a r_a =0$ for some coefficients $x_a \in \mathbb{F}_q$, then using~\Cref{eq:inner-prod} we find that $x_a=0$ for all $a \in H_1$. Similarly, suppose there exists $x=\sum_{a \in H_1} x_a r_a = \sum_{a \in H_0} x'_a r_a \in \rowspan(H_0)$, then~\Cref{eq:inner-prod} implies $x_a=0$ and hence $x=0$, i.e., $\rowspan(H_1) \cap \rowspan(H_0) = \{0\}$.
    
    Therefore, we can define the following quantum CSS code on $N=n-K$ qudits. Each logical Z codeword $u \in \mathbb{F}_q^K$ is
    \begin{align}
        \ket{\overline{u}} =  \frac{1}{\sqrt{2^{\dim H_0}}} \sum_{h \in \rowspan(H_0)} \ket{u H_1 + h}.
    \end{align}
    In other words, the quantum code encodes $K$ qudits. And it is CSS$(\mcC'', \mcC'^\perp)$ where
    \begin{align}
        \mcC'' = \rowspan(H_0,H_1), \qquad \mcC' =\rowspan(H_0).
        \label{eq:shortened-code}
    \end{align}
    We refer to the code $\mcC'$ as a $K$-shortened version of $\mcC$, in analogy with the shortened Reed-Solomon codes.
    
    Let us now consider the action of $(CCZ^{(q)})^{\otimes N}$ on $\ket{\overline{u}}\ket{\overline{v}}\ket{\overline{w}}$ (below we suppress the normalization factor and the sum notations for brevity)
    \begin{align*}
    (CCZ^{(q)})^{\otimes N}&\ket{\overline{u}}\ket{\overline{v}}\ket{\overline{w}} \\
    &\propto \sum_{h, h', h''} \prod_{i=1}^N (-1)^{\tr((\sum_{a}u_a r_a +h)_i(\sum_{b} v_b r_b +h')_i(\sum_{c} w_c r_c +h'')_i) }  \ket{u H_1 + h}\ket{v H_1 + h'} \ket{w H_1 + h''},\\
    &\propto \sum_{h, h', h''} (-1)^{\tr( \sum_{i=1}^N (\sum_{a}u_a r_a +h)_i(\sum_{b} v_b r_b +h')_i(\sum_{c} w_c r_c +h'')_i) }  \ket{u H_1 + h}\ket{v H_1 + h'} \ket{w H_1 + h''},
    \end{align*}
    where the sums in the exponent are over rows $a, b, c \in H_1$, and we use the $\mathbb{F}_2$-linearity of the trace map. Using~\Cref{eq:punctured-star} and the $\mathbb{F}_2$-linearity of the trace map, we can simplify the sum inside the trace map in the exponent to
    \begin{align}
        \sum_{i=1}^{N} \left(\sum_{a}u_a r_a +h \right)_i \left(\sum_{b} v_b r_b +h' \right)_i \left(\sum_{c} w_c r_c +h''\right)_i= \sum_{a \in H_1} u_a v_a w_a.
        \label{eq:333}
    \end{align}
    Hence,
    \begin{align*}
        (CCZ^{(q)})^{\otimes N}\ket{\overline{u}}\ket{\overline{v}} \ket{\overline{w}} & \propto \sum_{h, h', h''}  (-1)^{\sum_{a=1}^K \tr(u_a v_a w_a)} \ket{\overline{u}}\ket{\overline{v}} \ket{\overline{w}} \\ 
        &=(\overline{CCZ^{(q)}})^{\otimes K} \ket{\overline{u}}\ket{\overline{v}} \ket{\overline{w}}.
    \end{align*}

    Next, we determine the quantum code distance. As usual, the X distance is lower bounded by the distance of the classical code $\rowspan(H_0,H_1)$, which is at least $ d- K$ because we only puncture $K$ coordinates out of a distance-$d$ code $\mcC$. The Z distance, on the other hand, is lower bounded by the distance of the dual of the $K$-shortened code $\mcC' \triangleq \rowspan(H_0)$ and we do not in general have a lowerbound for it.

    Finally, observe that in the above derivation we only used the fact that $\tr(\cdot)$ is $\mbF_2$-linear. Furthermore, the argument of $\tr(\cdot)$ can be generalized to be any polynomial $g$ over $\mbF_q$ of degree at most $3$ (for the $CCZ$ gate $g(x,y,z)=xyz$). Indeed, this is because using~\Cref{eq:inner-prod} and~\Cref{eq:weight} we can obtain the following expressions for the degree-2 and degree-1 terms that are similar to~\Cref{eq:333}:
    \begin{align*}
        \sum_{i=1}^{N} \left(\sum_{a \in H_1}u_a r_a +h \right)_i \left(\sum_{b \in H_1} v_b r_b +h' \right)_i = \sum_{a \in H_1} u_a v_a,
    \end{align*}
    \begin{align*}
        \sum_{i=1}^{N} \left(\sum_{a \in H_1}u_a r_a +h \right)_i = \sum_{a \in H_1} u_a.
    \end{align*}
    Hence, any degree-3 phase gates of the form~\Cref{eq:low-degree-gate} can be implemented transversally.
\end{proof}

\begin{remark}
    We can readily see that the proof of the above lemma can be directly generalized to higher $t$-multiplication property, such that the constructed quantum code supports degree-$(t+1)$ phase gates. We do not spell out the details of this here.
\end{remark}

\section{Good classical code with multiplication property}\label{sec:AG}
We now instantiate~\Cref{lem:ccz} using a good classical code family with the multiplication property, whereby obtaining an asymptotically good qudit CSS code with transversal low-degree phase gates in~\Cref{cor:qudit-code}. The classical code we will use belongs to the class of algebraic-geometry (AG) codes, also known as geometric Goppa codes~\cite{goppa1981codes, tsfasman1982modular}. These codes are generalizations of both Reed-Solomon and Reed-Muller codes and are often described using the language of algebraic function fields. They are known to address the growing alphabet issue of RS code and the vanishing rate of RM code. We will sum up several facts needed to describe the code construction, but ignoring some technical details (following~\cite{cascudo2018amortized}). We refer the reader to the standard textbook by Stichtenoth~\cite{stichtenoth2009algebraic} for more details about this area.

\paragraph{Function fields.} Let $\mathbb{F}_q$ be a finite field. The rational function field $\mbF_q(x)$ is the field containing all rational functions with coefficients in $\mbF_q$. An \emph{algebraic function field} (or function field) $F/\mbF_q$ is an algebraic extension of the rational function field. Associated to a function field is a non-negative integer $g(F)$ called the \emph{genus}, and an infinite set of `places' $P$, each having a degree $\operatorname{deg} P \in \mathbb{N}$. The number of places of a given degree is finite. The places of degree 1 are called \emph{rational places}. Given a function $f \in F$ and a place $P$, two things can happen: either $f$ has a pole at $P$, or $f$ can be evaluated at $P$ and the evaluation $f(P)$ can be seen as an element of the field $\mathbb{F}_{q^{\deg P}}$. If $f$ and $g$ do not have a pole at $P$, then the evaluations satisfy the rules $\lambda(f(P))=(\lambda f)(P)$ for every $\lambda \in \mathbb{F}_q$, $f(P)+g(P)=(f+g)(P)$, and $f(P) g(P) = (fg)(P)$. So if $P$ is rational and $f$ does not have a pole at $P$, then $f(P) \in \mbF_q$.

An important fact of the theory of algebraic function fields is as follows. Let $N(F)$ be the number of rational places of $F$. Then over every finite field $\mbF_q$, there exists an infinite family of function fields $\{F_i\}$ such that their genus $g(F_i)$ grows with $i$ and $\lim \frac{N(F_i)}{g(F_i)} = c_q > 0$. The largest constant $c_q$ satisfying this property (over all function field families on $\mathbb{F}_q$) is called the Ihara's constant $A(q)$ of $\mbF_q$~\cite{ihara1981some}. It is known that $A(q) \leq \sqrt{q}-1$ (Drinfeld-Vladut bound~\cite{vladut1983number}) for every finite field. When $q$ is a square, there are explicit families of funcion fields attaining $A(q)=\sqrt{q}-1$, first constructed by~\cite{garcia1995tower}. Explicit constructions in the general non-square case are given in~\cite{bassa2012towers}.

\paragraph{Divisors and Riemann-Roch space.} A \emph{divisor} $G$ is a formal sum of places, $G = \sum c_P P$ such that $c_P \in \mathbb{Z}$ and $c_P=0$ except for a finite number of $P$. The set of places where $c_P \neq 0$ is called the support, $\operatorname{supp}(G)$, of $G$. The degree of $G$ is defined to be $\deg(G) = \sum c_P \deg(P) \in \mathbb{Z}$. The divisor is said to be positive if $c_P \geq 0$ at all places, denoted as $G \geq 0$. For two divisors $G_1$ and $G_2$, we write $G_1 \geq G_2$ if $G_1 - G_2 \geq 0$.
Given a divisor $G$, one can define the \emph{Riemann-Roch space} $L(G)$, which is an $\mbF_q$-vector space of all functions in $F$ with certain prescribed poles and zeros depending on $G$. More specifically, every function $f \in L(G)$ must have a zero of order at least $|c_P|$ at the places $P$ with $c_P <0$ and $f$ can have a pole of order at most $c_P$ at the places $P$ with $c_P >0$. By definition it is easy to see that, given $f,g \in L(G)$, their product $fg$ is in $L(2G)$. The dimension of $L(G)$ is denoted as $\ell(G)= \operatorname{dim} L(G)$ and is governed by the Riemann-Roch theorem. For most coding theory purposes, it suffices to use the following version.

\begin{theorem}[Riemann-Roch]\label{lem:riemann} Let $G$ be a divisor over a function field $F$. It holds that $\ell(G)\geq \deg (G)-g(F) +1$,
and the equality holds if $\deg (G) \geq 2g(F) -1$. On the other hand, if $\operatorname{deg} (G) < 0$, then $\ell(G)=0$.
\end{theorem}

For a set of rational places $\mcP=(P_1,\hdots,P_n)$ of $F$, we  often denote $D_\mcP$ to be the divisor $D_\mcP = \sum_{i=1}^n P_n$. In this section, we will always consider the case when $\operatorname{supp}(D_\mcP) \cap \operatorname{supp}(G) = \emptyset$. We can define an $\mbF_q$-linear evaluation map
\begin{align*}
    \mathsf{ev}_{D_\mcP, G}:\qquad  L(G) &\mapsto \mathbb{F}_q^n\\
     f &\mapsto (f(P_1),\hdots,f(P_n)),
\end{align*}

\begin{definition}[AG code] The linear code $C_L(D_\mcP, G)$ is the image of the Riemann-Roch space $L(G)$ under the evaluation map $\mathsf{ev}_{D_\mcP, G}$.
\end{definition}

The Riemann-Roch theorem provides us with a tool to bound AG code parameters.

\begin{lemma}[AG code parameters]
\label{lem:agcode-params} Suppose $ \deg(G) < n$. The code $C_L(D_\mcP, G)$ is a linear code of length $n= |\mcP|= \deg(D_\mcP)$ and dimension $k \geq \deg(G) +1 - g(F)$.  The equality $k = \deg(G) +1 - g(F)$ holds when $2g(F)-1\leq \deg(G) < n$. Furthermore, its distance is $d \geq n - \deg(G)$.
\end{lemma}
\begin{proof} The assumption $ \deg(G) < n$ implies that $\mathsf{ev}_{D_\mcP, G}$ is injective. Indeed, suppose there exists a nonzero function $f \in F$ that gets mapped to the all-0's codeword, then by definition we have $f \in L(G-D_\mcP)$. However, $\operatorname{deg}(G-D_\mcP)=\deg(G) - n  <0$ and thus~\Cref{lem:riemann} asserts that $\ell(G-D_\mcP)=0$. So no such $f$ exists. Hence, the code dimension $k=\ell(G)$ and follows the prescription in~\Cref{lem:riemann}.

The distance follows from a similar argument. Let $f$ be a nonzero function that gets mapped to the minimal weight codeword. Wlog, assume $f(P_i)=0$ for $i=1,..,z$ for some $z < n$. Then $f \in L(G-\sum_{i=1}^z P_i)$. By~\Cref{lem:riemann}, we require $\operatorname{deg}(G) - z \geq 0$ for such $f$ to exist. So the code distance is $d = n -z \geq n -\operatorname{deg}(G)$.
\end{proof}

From the code parameter bounds, we see that the goal is to find function fields with as many rational places as possible while keeping the genus low, hence the definition of Ihara's constant. Next, we note the following fact reflecting the behavior of the order of a pole/zero when two functions are multiplied.

\begin{fact}\label{fact:mult} For divisors $G \leq G'$, $C_L(D_\mcP, G) \subseteq C_L(D_\mcP, G')$. In addition $C_L(D_\mcP, G)  \star C_L(D_\mcP, G') \subseteq C_L(D_\mcP, G+G')$. As a corollary, $C_L(D_\mcP,G)^{\star t} \subseteq C_L(D_\mcP,tG)$.
\end{fact}

Similar to RS codes, the dual of an AG code is another AG code.

\begin{lemma}[Theorem 2.2.7 and Proposition 2.2.10 in~\cite{stichtenoth2009algebraic}]
\label{lem:dual}
    The dual of the code $C_L(D_\mcP, G)$ is $C_L(D_\mcP, R + D_\mcP -G)$, where $R$ is a suitable canonical divisor of $F$. Furthermore, if $2g(F)-1\leq \deg(G) < n$ then $C_L(D_\mcP, G)^\perp$ has dimension $k^\perp=n+g(F)-1-\deg(G)$ and distance $d^\perp \geq \deg (G) +2 -2 g(F)$.
\end{lemma}

\paragraph{Intuition for AG codes.} We briefly discuss why AG codes can improve over Reed-Solomon and Reed-Muller codes. The most intuitive way to think about AG codes is as evaluations of \emph{rational functions} on an \emph{algebraic curve} (i.e., the common roots of a collection of multivariate polynomials). In the above we have defined them using the terminology of function fields following~\cite{stichtenoth2009algebraic} since it's easier to formally write down proofs and cite the results from there that way (as we will also do below in~\Cref{thm:AG} and in~\Cref{sec:qubit}) -- but the algebraic and geometric pictures are equivalent~\cite[Section 1.6]{hartshorne2013algebraic}. Recall that RS code is defined by evaluations of polynomials over $\mathbb{F}_q$. We can think of RS as an AG code whose evaluation points are defined on a `line', with $q$ points corresponding to $\mathbb{F}_q$. RM code is a $t$-variate version of RS codes, whose evaluation points are on the `hypercube' $\mathbb{F}_q^t$. AG codes generalize RS and RM codes in the sense that we can be somewhere in between -- in this case, a curve\footnote{One may find an analogy between this discussion and the recent development from hypergraph product (HP) codes~\cite{tillich2013quantum} to balanced product (BP) codes~\cite{breuckmann2021balanced, panteleev2022asymptotically}. HP is a simple Cartesian product that undesirably involves too many qubits, and BP is a quotienting operation that reduces the number of physical qubits. For AG codes, defining an algebraic curve as the roots of a set of multivariate polynomials is in a sense quotienting the hypercube. However, we are not claiming any formal connection between the two subjects here.}.
First, we want the curve to be somewhat complicated so that we can have a large block length (i.e. many evaluation points/rational places) as we add more variables. Second, the prescriptions for pole and zero orders according to the divisor $G$ generalizes the degree constraint in RS and RM codes, explaining the familiar form of $n - \operatorname{deg}(G)$ for the AG code distance bound that is similar to RS codes. Third, the message functions are now additionally required to be evaluatable on the curve; the more complicated the curve is the more stringent this requirement is. This is reflected in the negative term in the code dimension related to the genus in~\Cref{lem:agcode-params}. Hence, we see a tension in the choice of the curve parameters in order to get a good code. The notion of genus and the Riemann-Roch theorem are a formalization of this discussion.

We now describe an AG code with the multiplication property from~\Cref{def:mult-prop}.
The following is an adaptation of the proof of Theorem 8.4.9 in~\cite{stichtenoth2009algebraic}, who derived a self-dual AG code. Here, we instead instantiate the parameters to obtain the multiplication property. 
Alternatively, we believe that other existing constructions of self-dual (or self-orthogonal) AG codes such as in~\cite{stichtenoth1988self, driencourt1989criterion, stichtenoth2006transitive} can also be adapted to obtain the multiplication property.

\begin{theorem}\label{thm:AG} For any even $m \geq 8$, there exists an explicit asymptotically good code family $\mcC$ over $\mathbb{F}_{2^m}$ with parameters $[n, \geq n/4, \geq n/4]_{\mathbb{F}_{2^m}}$ satisfying the multiplication property. Furthermore, it contains the all-1's codeword, and for $K < n/5$, its $K$-shortened version $\mcC'$ (defined in~\Cref{eq:shortened-code}) has the dual code with distance at least $n/5-K$.
\end{theorem}

\begin{proof}
    Let $q=\ell^2$, where $\ell$ is a power of 2.
    Theorem 7.4.15 and Corollary 7.4.16 in~\cite{stichtenoth2009algebraic} (which in turn builds from~\cite{garcia1995tower}) states that there exists an explicit infinite sequence $\mcF=(F_i)_{i \in \mathbb{N}^+}$ of function fields over $\mathbb{F}_q$ such that $F_0= \mathbb{F}_q(x)$. For $i \geq 1$, the number of rational places grows as $n_i=\ell^{O(i)}$ and the genus of $F_i$ is
\begin{align}
    g(F_i) = 1 + \frac{n_i}{\ell-1} \left(1 - \frac{1}{\alpha_i} - \frac{1}{\beta_i}\right),
\end{align}
where $\alpha_i, \beta_i \geq 2$ are even integers that increase with $i$, whose details are not too important for us. Let the divisor corresponding to the $n_i$ rational points be $D_i =\sum_{P \in \mcP_i} P$. Moreover, according to Equations 8.17 and 8.18 in~\cite{stichtenoth2009algebraic}, there exists a choice of canonical divisor $R_i$ as in~\Cref{lem:dual}, that has the form
\begin{align}
    R_{i} = (\ell \alpha_i - 2)A_i + (\beta_i -2)B_i - D_{i},
\end{align}
where $A_i, B_i$ are positive divisors whose degrees satisfy $\alpha_i \operatorname{deg}(A_i) = \beta_i \operatorname{deg}(B_i)=n_i/(\ell-1)$ and $(\operatorname{supp} A_i \cup \operatorname{supp} B_i) \cap \supp D_i=\emptyset$.

To apply this construction, we need to choose a divisor $G_i$ for each $i$, giving an AG function code family $\mcC_i = C_L(D_i, G_i)$ with dual code $\mcC_i^\perp = C_L(D_i, R_i + D_i -G_i)$. Since $\mcC_i \star \mcC_i \subseteq C_L(D_i, 2 G_i)$ (\Cref{fact:mult}), it suffices to choose $G_i$ such that $2G_i \leq R_i +D_i -G_i$, to ensure $\mcC_i \star \mcC_i \subseteq \mcC_i^\perp$. Rearranging this to $3G_i \leq (\ell \alpha_i - 2)A_i + (\beta_i -2)B_i$, we find that $G_i$ can be chosen to be $G_i = s A_i$ where $s=\lfloor(\ell \alpha_i -2)/3\rfloor > 0$ as long as $\ell \geq 4$ (we are using the fact that both $A_i, B_i$ are positive divisors).

With the above choice, according to~\Cref{lem:agcode-params}, the dimension of $\mcC_i$ is
\begin{align}
    k_i &\geq \operatorname{deg} G_i + 1 -g(F_i)\\
    &\geq \lfloor(\ell \alpha_i -2)/3\rfloor \operatorname{deg} A_i - \frac{n_i}{\ell-1} \left(1 - \frac{1}{\alpha_i} - \frac{1}{\beta_i}\right)\label{eq:1212}\\
    &\geq \frac{n_i}{\ell-1}(\lfloor\ell/3\rfloor  - 1). \tag{using $\alpha_i \operatorname{deg}(A_i) =n_i/(\ell-1)$}
\end{align}

Choosing $\ell \geq 16$ ensures $k_i/n_i \geq 1/4$. Similarly, for $\ell \geq 16$, the distance of $\mcC_i$ is lowerbounded by
\begin{align}
    d_i &\geq n_i - \operatorname{deg}(G_i)\\
    &\geq n_i - \lfloor(\ell \alpha_i -2)/3\rfloor \operatorname{deg} A_i\\
    &\geq n_i\left(1-\frac{2\ell}{3(\ell-1)}\right) \geq n_i/4.
\end{align}

Since the divisor $G_i$ is positive, i.e., it only requires there aren't too many poles at each point in $\operatorname{supp}(G_i)$, the constant function $1$, which don't have any poles and zeroes, is in $L(G_i)$. In other words, $\mcC_i$ contains the all-1's codeword.

Finally, we lowerbound the distance of the dual of the $K$-shortened code $\mcC'_i$ of $\mcC_i$, defined in~\Cref{eq:shortened-code}. Wlog we assume the punctured coordinates correspond to the first $K$ of the rational points, as visualized in~\Cref{eq:puncture} and repeated below for convenience:
    \begin{align*}
        \begin{pmatrix}
            \mathbf{1}_K & H_1\\
            0 & H_0
        \end{pmatrix},
    \end{align*}
i.e., $\mcC_i' =\rowspan(H_0)$.
Observe that, this $K$-shortened code is another AG code, $C_L(D'_i, G'_i)$, where $D'_i = \sum_{i=K+1}^{n_i} P_i 
 \leq D_i$ and $G'_i = G_i - \sum_{i=1}^{K} P_i$. This is because $\rowspan(H_0)$ is a restriction of the elements within $L(G_i)$ to those that have at least one zeroes on each of the points $P_1,\hdots,P_K$. Since $ \operatorname{supp}(G_i) \cap \operatorname{supp}(D_i) = \emptyset$, such elements are specified exactly by the divisor $G'_i = G_i - \sum_{i=1}^{K} P_i $. Now, using~\Cref{lem:dual}, we obtain a lowerbound on the distance of $(\mcC'_i)^\perp$ as 
 \begin{align}
     \operatorname{dist}((\mcC'_i)^\perp) &\geq \deg(G'_i) +2 -2g(F_i)\\
     &\geq \lfloor(\ell \alpha_i -2)/3\rfloor \operatorname{deg} A_i - K + 2 - 2 \frac{n_i}{\ell-1} \left(1 - \frac{1}{\alpha_i} - \frac{1}{\beta_i}\right).\\
     & \geq \frac{n_i}{\ell-1}(\lfloor\ell/3\rfloor -2 ) +2 -K,
 \end{align}
 hence choosing $\ell\geq 16$ ensures $\operatorname{dist}((\mcC'_i)^\perp) \geq n_i/5 - K$.
\end{proof}

Substituting the code from~\Cref{thm:AG} into the construction in~\Cref{lem:ccz}, we obtain the following corollary.

\begin{cor}\label{cor:qudit-code} There exists an explicit qudit CSS code family of parameters $[[n - K, K, \geq n/5 - K ]]_{\mathbb{F}_{q=2^m}}$ for any $m \geq 8$, where we can take any value $K < n/5$. Furthermore, the logical operation $\overline{CCZ^{(q)}}^{\otimes K}$ can be realized  by transversal $(CCZ^{(q)})^{\otimes N}$. The same transversality statement also holds for the degree-3 phase gates defined in~\Cref{eq:low-degree-gate}.
\end{cor}
Here, we note that the transversality of other degree-3 phase gates will be used crucially in the proof of the main result in the next section.

\begin{remark}
    By inspecting the choice of the divisor $G_i$ in the proof above, we can see that it can be adapted to support higher $t$-multiplication property (at the cost of worse rate/distance). However, $t$ is limited by the field size as $t < O(\sqrt{q})$ due to, for example, the rate lowerbound in~\Cref{eq:1212} becomes negative.
\end{remark}

\section{Qudit-to-qubit code transformation}\label{sec:qubit}
In this section, we derive the main result of this paper (\Cref{thm:main}).
To obtain a qubit CSS code with transversal $CCZ$, we first look at how to convert the qudit CSS code in~\Cref{cor:qudit-code} into a qubit code, without yet worrying about the $CCZ$ gate. This question has been addressed elsewhere (see, e.g.,~\cite[Theorem 4]{grassl1999quantum} or~\cite[Section 8]{cross2008}) using \emph{self-dual bases}, which always exist for extension fields of characteristic 2~\cite{gao1993normal}. We here summarize this result and give a proof for completeness. A basis for the field $\mathbb{F}_{2^m}$ is a set of linearly independent elements $\mcB=\{\alpha_1,\hdots,\alpha_{m}\}$ such that any field element $u$ can be expressed as $x=\sum_{i} x_i \alpha_i$ where $x_i \in \mathbb{F}_2$. The basis is self-dual if $\tr(\alpha_i \alpha_j) = \delta_{ij}$. We denote by $\mcB(x) \in \mathbb{F}_2^m$ the basis expansion of an element $x$.

\begin{lemma}\label{lem:qubit-selfdual} Suppose we are given an $\mathbb{F}_{q=2^m}$-qudit CSS code $\mcQ$ of parameters $[[N,K, \geq D]]_{\mathbb{F}_{2^m}}$ based on two classical $\mathbb{F}_{2^m}$ codes $\mcC_1$ and $\mcC_2$ with $\mcC_2^\perp \subseteq \mcC_1$ and $D= \min(\operatorname{dist}(\mcC_1),\operatorname{dist}(\mcC_2))$. Then, under any self-dual basis mapping $\mathbb{F}_{2^m} \ni a \mapsto \mcB(a) \in \mathbb{F}_2^m$, we obtain a qubit CSS code with parameters $[[Nm,Km, \geq D]]$ using the classical binary codes $\mcB(\mcC_1)$ and $\mcB(\mcC_2)$. Furthermore, the qudit Paulis naturally transform as $Z^{(q)}_a \mapsto Z(\mcB(a)) \triangleq Z^{\mcB(a)_1} \otimes \hdots \otimes Z^{\mcB(a)_m}$ and similarly for the $X$ Paulis.
\end{lemma}
\begin{proof}
    We need three facts. The first fact is obvious and holds for any basis $\mcB$ of $\mathbb{F}_{q=2^m}$: if $\mcC \subseteq \mcC' \subseteq \mathbb{F}_q^N$ are two linear codes then $\mcB(\mcC) \subseteq \mcB(\mcC') \subseteq (\mathbb{F}_2^{m})^N$. The second fact holds when $\mcB$ is self-dual: we have $\mcB(\mcC)^\perp = \mcB(\mcC^\perp)$. Indeed, consider any two codewords $a \in (a_1,\hdots,a_N)\in \mcC$ and $b= (b_1,\hdots,b_N) \in\mcC^\perp$. We write their basis expansion as $b_i = \sum_{j=1}^m b_{i,j} \alpha_j$ and similarly for $a_i$. Because $|a \star b|=0$ over $\mathbb{F}_q$, it holds that
    \begin{align*}
        0&=\tr\left(\sum_{i=1}^N  a_i b_i\right)=\sum_{i=1}^N\sum_{j,l=1}^m a_{i,j} b_{i,l} \tr\left(  \alpha_{j} \alpha_l\right)\\
        &= \sum_{i=1}^N\sum_{j=1}^m a_{i,j} b_{i,j}= |\mcB(a)\star \mcB(b)|.
    \end{align*}
    This means $\mcB(\mcC^\perp) \subseteq \mcB(\mcC)^\perp$. This is in fact an equality because both are linear subspaces of $\mathbb{F}_2^{Nm}$ of dimension $m(N-\operatorname{dim}_{\mathbb{F}_q}(\mcC))$. The third fact is $\operatorname{dist}(\mcB(\mcC)) \geq \operatorname{dist}(\mcC)$.

    Given a quantum code CSS$(\mcC_1,\mcC_2)$ over $\mathbb{F}_q$, we can straightforwardly use the above facts to obtain a qubit code. Let $\mcC_1'=\mcB(\mcC_1)$ and $\mcC_2'=\mcB(\mcC_2)$ be the binary codes expanded in the self-dual basis. Then we have $\mcC_2'^\perp \subseteq \mcC_1'$ due to facts 1 and 2. Hence, the qubit CSS code CSS$(\mcC_1',\mcC_2')$ is well-defined. The number of physic qubits and logical qubits can be easily seen to be $Nm$ and $Km$, respectively. The qubit code distance is lowerbouned by $\min(\operatorname{dist}(\mcC_1'),\operatorname{dist}(\mcC_2')) \geq \min(\operatorname{dist}(\mcC_1),\operatorname{dist}(\mcC_2)) \geq D$ due to fact 3.

    Finally, the expansion rules for qudit Paulis can also be seen to be well-defined. Let $u, v \in \mathbb{F}_q$ and consider $Z^{(q)}_u \mapsto Z(\mcB(u)) \triangleq Z^{\mcB(u)_1} \otimes \hdots \otimes Z^{\mcB(u)_m}$ and $X^{(q)}_v \mapsto X(\mcB(v)) \triangleq X^{\mcB(v)_1} \otimes \hdots \otimes X^{\mcB(v)_m}$, then we have
    \begin{align*}
        Z(\mcB(u))X(\mcB(v)) &= (-1)^{\sum_{i=1}^m \mcB(u)_i\mcB(v)_i}X(\mcB(v)) Z(\mcB(u))\\
        &=(-1)^{\tr(uv)}X(\mcB(v)) Z(\mcB(u)), \tag{using self-duality}
    \end{align*}
    satisfying the qudit Pauli commutation relation.
\end{proof}

To convert the code into binary alphabet without losing the transversality of the (now) qubit CCZ, we use known techniques from the linear secret sharing scheme literature. First, we define the notion of reverse multiplication-friendly embedding that was defined and constructed in~\cite{cascudo2018amortized, escudero2023degree}. 

\begin{definition}[Reverse multiplication-friendly embedding (RMFE)]\label{def:rmfe}
A degree-$t$ RMFE for the vector space $\mathbb{F}_2^s$ is a triple $(m,\phi_\rmfe,\psi_\rmfe)$ where $m$ is a positive integer and where $\phi_\rmfe: \mathbb{F}_2^{s} \mapsto \mathbb{F}_{2^m}$ and $\psi_\rmfe: \mathbb{F}_{2^m} \mapsto \mathbb{F}_{2}^{s}$ are (injective and surjective, respectively) $\mathbb{F}_2$-linear maps, such that $x_1 \star x_2\star...\star x_t = \psi_\rmfe (\phi_\rmfe(x_1) ...  \phi_\rmfe(x_t))$ for all $x_1,...,x_t \in \mathbb{F}_{2}^s$.
\end{definition}

Interestingly, RMFE can be constructed using algebraic function fields. Here, we give a degree-3 generalization of the degree-2 RMFE in~\cite{cascudo2018amortized}. This construction is non-explicit, but it will suffice for our purposes since we will apply it with $s, m \in O(1)$. (In fact, a 3-RMFE with $s=1$ -- e.g., the trivial RMFE via identifying $\mathbb{F}_2$ with the elements 0, 1 of $\mathbb{F}_{2^m}$ -- will also suffice to obtain the main asymptotic behavior in~\Cref{thm:main}. However, we believe the construction below is of independent interest. It allows us to extract more logical qubits out of a qudit code -- see step 2 of~\Cref{thm:main}'s proof.)

\begin{lemma}\label{lem:rmfe} There exists a degree-3 RMFE with $m=\Theta(s)$ for infinitely many sufficiently large $s$.
\end{lemma}
\begin{proof} We follow Section 4 of~\cite{cascudo2018amortized}. Let $F/\mbF_2$ be a function field of genus $g$ with $s$ rational places $P_1,\hdots,P_s$. According to~\cite{xing2007algebraic}, there exists a function field family over $\mathbb{F}_2$ with Ihara's constant $A(2) \geq 97/376$, i.e., the ratio $s/g  >0$, and approaches a constant $\geq 97/376$. According to Theorem 5.2.10(c) in~\cite{stichtenoth2009algebraic}, for every function field $F/\mathbb{F}_q$ and all $m \in \mathbb{N}$ with $2g + 1 \leq q^{(m-1)/2}(\sqrt{q}-1)$, there exists a place in $F$ of degree $m$. So in particular, this holds for every $m \geq 4g+3$, regardless of $g$ and $q$.

Let $G$ be a divisor of $F$ with degree $\deg(G) = s+2g-1$ and such that $\supp(G) \cap \{P_1,\hdots,P_s\} = \emptyset$. Since $\deg(G) \geq 2g -1$,~\Cref{lem:riemann} states that $\ell(G) = \deg(G) - g +1 $. Similarly $\ell(G- \sum_{i=1}^s P_i)= \deg(G) - s -g+1$. Take a place $R$ of degree $m = 3s + 6g -2  > 3 \deg(G)$.

Consider the ($\mathbb{F}_2$-linear) evaluation map $\pi: L(G) \mapsto \mathbb{F}_2^s$ defined by $\pi(f) = (f(P_1),\hdots,f(P_s))$. Then the kernel of $\pi$ is $L(G- \sum_{i=1}^s P_i) $. Since $\dim_{\mathbb{F}_2} \Im(\pi) = \ell(G)- \ell(G- \sum_{i=1}^s P_i) = s$, $\pi$ is surjective. So we can choose a linear subspace $W$ of $L(G)$ of dimension $s$ such that $\pi$ induces an isomorphism between $W$ and $\mathbb{F}_2^s$. The RMFE embedding map is defined to be, for each $f \in W$,
\begin{align}
    \phi: \mathbb{F}_2^s \mapsto \mathbb{F}_{2^m} \qquad   \phi(\pi(f)) = f(R).
\end{align}
This definition makes sense because $\pi$ is onto $\mathbb{F}_2^s$. And we recall from~\Cref{sec:AG} that the function evaluation at a place of degree $m$ takes values in $\mbF_{2^m}$.
We need to show that $\phi$ is $\mathbb{F}_2$-linear and injective. Indeed, $\phi $ is $\mbF_2$-linear since $\pi$ is a $\mathbb{F}_2$-linear isomorphism between $W$ and $\mathbb{F}_2^s$ and the addition rule $(f+f')(R)= f(R)+ f'(R)$. Furthermore, that $\phi$ is injective follows from $\deg(R)=m > \deg(G)$. Indeed, suppose for contradiction that $\phi$ is not injective, i.e., there exist $f \neq f' \in W \subset L(G)$ such that $f(R)=f'(R)=0$. This implies $(f+f')(R)=0$ and hence $f+f' \in L(G - R)$. But $\deg(G-R) =  \deg(G) - m <0$, and thus according to~\Cref{lem:riemann}, the Riemann-Roch space $L(G - R) = \{0\}$, contradicting with the assumption $f \neq f'$.

Next we define the un-embedding map. First, define $\tau: L(3G) \mapsto \mbF_{2^m}$ as $\tau(f) = f(R)$. Similarly to the argument above with $\phi$, we can show that $\tau$ is $\mbF_2$-linear and injective since $m > \deg(3G)$. We define the map $\psi': \Im(\tau) \subseteq \mathbb{F}_{2^m} \mapsto \mathbb{F}_{2}^s$ as follows, $\psi'(f(R)) = (f(P_1),\hdots,f(P_s))$. This map is well-defined and is $\mathbb{F}_2$-linear because $\tau$ is $\mathbb{F}_2$-linear and injective. It is also surjective by a similar argument to why $\pi$ is surjective above. Finally, we linearly extend $\psi'$ from $\Im(\tau)$ to all of $\mbF_{2^m}$ to obtain the un-embedding map $\psi $. Let us verify that $(\phi, \psi )$ forms a degree-3 RMFE. For any $f_1, f_2, f_3 \in W$, we have
\begin{align*}
    \psi (\phi (\pi(f_1))  \phi (\pi(f_2)) \phi (\pi(f_3)) ) &=\psi  (f_1(R)f_2(R)f_3(R)) \\
    &=\psi  ((f_1f_2f_3)(R)) \\
    &=((f_1f_2f_3)(P_1),\hdots,(f_1f_2f_3)(P_s)) \\
    &= \pi(f_1) \star \pi(f_2) \star\pi(f_3),
\end{align*}
where we used that $(f_1f_2f_3)(R) \in \Im(\tau)$ since $f_1f_2f_3 \in L(3G)$.
\end{proof}

Next, we give the notion of multiplication-friendly embedding (MFE), which was historically defined before RMFE~\cite{cascudo2009asymptotically}. We generalize the definition in~\cite{cascudo2009asymptotically} to support more multiplications.

\begin{definition}[Multiplication-friendly embedding (MFE)]\label{def:mfe}
A degree-$t$ MFE of the extension field $\mathbb{F}_{2^m}$ over $\mathbb{F}_2$ is a triple $(r,\sigma_\mfe,\psi_\mfe)$ where $r$ is a positive integer and where $\sigma_\mfe: \mathbb{F}_{2^m} \mapsto \mathbb{F}_2^r$ and $\psi_\mfe: \mathbb{F}_2^r \mapsto \mathbb{F}_{2^m}$ are (injective and surjective, respectively) $\mathbb{F}_2$-linear maps such that the following holds. There exist $r$-element permutations $\pi_{2},\hdots,\pi_{t}$, such that for all $x_1,...,x_t \in \mathbb{F}_{2^m}$ it holds that $x_1 x_2...x_t = \psi_\mfe (\sigma_\mfe(x_1) \star \pi_2(\sigma_\mfe(x_2)) \star ...\star  \pi_t(\sigma_\mfe(x_t)))$.
\end{definition}

The construction~\footnote{It is natural to ask why this definition involves the extra permutations as compared to RMFE. In fact, in the first version of this paper, we incorrectly claimed that degree-3 MFE was possible without the permutations by generalizing a degree-2 MFE in~\cite[Theorem 9]{cascudo2009asymptotically}. However, Louis Golowich pointed out to us that there is an impossibility argument invalidating our claimed proof. We are extremely grateful for his pointer.} we give below is explicit, but asymptotically bad, i.e., $m/r \rightarrow 0$. However, it suffices for our purposes as we will apply it with $m=O(1)$.

\begin{lemma}\label{lem:mfe}
There exists a degree-3 MFE with $r= m^3$ for any $m$.
\end{lemma}
\begin{proof}
    We follow the idea of Theorem 9 in~\cite{cascudo2009asymptotically} by using polynomial bases of finite fields~\cite{mullen2013handbook}.
    Let $\alpha \in \mathbb{F}_{2^m}$ be a primitive element, i.e., $\alpha$ is a root of a degree-$m$ irreducible polynomial over $\mathbb{F}_2$, such that $\{1,\alpha,\hdots,\alpha^{m-1}\}$ is a basis for $ \mathbb{F}_{2^m}$. An element $x$ is decomposed as $x= \sum_{i=0}^{m-1} x_i \alpha^i$. Now, we define the embedding map as $\sigma(x) = (\underbrace{x_0,\hdots,x_0}_{m^2 \text{ times}}, \hdots \underbrace{x_{m-1}, \hdots, x_{m-1} }_{m^2 \text{ times}})$. Clearly this map is $\mathbb{F}_2$-linear and injective. We index the components of $\sigma(x)$ by $(\sigma(x))_{i,(j,k)}$ for $0 \leq i,j, k \leq m-1$, where $i$ labels which group of components and $(j,k)$ labels the elements in that group. In other words, $(\sigma(x))_{i,(j,k)} = x_i$ for all $j,k$.

    Given three field elements, $x,y,z \in \mathbb{F}_{2^m}$, we have $xyz= \sum_{i,j,k=0}^{m-1}  x_i y_j z_k \alpha^{i+j+k} $.  Observe that $x_i y_j z_k=(\sigma(x))_{i,(j,k)}  (\sigma(y))_{j,(i,k)} (\sigma(z))_{k,(i,j)}$. Hence, there exist $m^3$-element permutations $\pi$ and $\pi'$ such that
    $(\sigma(x) \star \pi(\sigma(y)) \star \pi'(\sigma(z)))_{i,(j,k)}=x_iy_jz_k$. It follows that the $\mathbb{F}_2$-linear map $\psi: \mathbb{F}_2^{m^3} \mapsto \mathbb{F}_{2^m}$ defined by $\psi(u) \triangleq \sum_{i,j,k} u_{i,(j,k)} \alpha^{i+j+k}$  satisfies $\psi(\sigma(x)\star \pi(\sigma(y)) \star \pi'(\sigma(z)))=xyz$. Finally, $\psi$ can be seen to be surjective by taking $y,z$ to be the identity element and varying $x$ over all elements of $\mathbb{F}_{2^m}$.
\end{proof}

We now use the MFE and RMFE constructed above to prove the main result. The idea of this construction is as follows. Consider the qudit code from~\cref{cor:qudit-code} on which physical qudit $CCZ$ gates realize logical qudit $CCZ$ gates. The RMFE is used to deal with the `logical end', i.e., we will restrict the logical space on each qudit to a subspace that simulates qubits. The MFE, on the other hand, handles the `physical end'. It turns out they interact nicely thanks to the fact that both are $\mathbb{F}_2$-linear and the crucial fact that the qudit code supports general degree-3 phase gates from~\Cref{eq:low-degree-gate} (not just the qudit $CCZ$).

\begin{theorem}[Main result]
\label{thm:main}
There exists an explicit quantum CSS code family on qubits with parameters $[[N, K=\Theta(N), D=\Theta(N)]]$ on which the logical gate $\overline{CCZ}^{\otimes K}$ gate is realized by a transversal application of CCZ gates on a constant fraction of the physical qubits.
\end{theorem}
\begin{proof}
    We take the $\mathbb{F}_{q=2^m}$-qudit code family from~\Cref{cor:qudit-code} and choose $m$ (a constant) such that there exists a degree-3 RMFE with corresponding parameter $m$. Let us refer to this $\mathbb{F}_{2^m}$-qudit code family as $\mcQ_0$ with parameters $[[N_0,K_0=\Theta(N_0),D_0=\Theta(N_0)]]_{\mathbb{F}_{2^m}}$. We consider 3-qudit phase gates of the form
    \begin{align}
        U^{(q)}_f \ket{x}\ket{y}\ket{z} = (-1)^{f(xyz)} \ket{x}\ket{y}\ket{z}, \qquad x,y,z \in \mathbb{F}_{2^m},
    \end{align}
    for some $\mathbb{F}_2$-linear function $f: \mathbb{F}_{2^m} \mapsto \mathbb{F}_2$, such as the trace map, but later we crucially will take it to be a function relevant to $\psi_\rmfe$. According to~\Cref{cor:qudit-code}, $(U^{(q)}_f)^{\otimes N_0}$ implements $ (\overline{U^{(q)}_f})^{\otimes K_0}$ on $\mcQ_0$.
    
    We will convert $\mcQ_0$ into the claimed qubit code in multiple steps and keep track of how the code parameters and the qudit transversal gate $(U^{(q)}_f)^{\otimes N_0}=(\overline{U^{(q)}_f})^{\otimes K_0}$ change.
    
    \paragraph{Step 1:} We apply the mapping from~\Cref{lem:qubit-selfdual} to obtain a qubit CSS code $\mcQ_1$ with parameters $[[N_1=N_0m, K_1=K_0m, D_1 \geq D_0]]$. Under this mapping, each $\mathbb{F}_{2^m}$-qudit (physical or logical) becomes a group of $m$ qubits. For later use, let us denote the logical Pauli Z operator corresponding to logical qubit $j \in [m]$ in group $i \in [K_0]$ as $\overline{Z}_{i,j}$. Below, for $\bs u=(u_1,\hdots,u_{N_0}) \in \mathbb{F}^{N_0}_{2^m}$, we use the notation $\mcB(\bs u) \triangleq (\mcB(u_1),\hdots ,\mcB(u_{N_0}) \in (\mathbb{F}_2^m)^{N_0})$, and similarly when we replace $N_0$ by $K_0$.
    
    Let us see how $U_f^{(q)}$ transforms. The following diagonal $3m$-qubit unitary effects $U^{(q)}_f$
    \begin{align}
        V_f \ket{\mcB(x)} \ket{\mcB(y)} \ket{\mcB(z)} \triangleq (-1)^{f(xyz)} \ket{\mcB(x)} \ket{\mcB(y)} \ket{\mcB(z)}, \qquad \forall x,y,z \in \mathbb{F}_{2^m}.
    \end{align}
    We will concern about how to implement $V_f$ later. For now, we observe that this gate is (block-) transversal on $\mcQ_1$. For any $\bs u, \bs v, \bs w \in \mathbb{F}_{2^m}^{K_0}$,
    \begin{align}
        V_f^{\otimes N_0} \ket{\overline{\mcB(\bs u)}}_{\mcQ_1} \ket{\overline{\mcB(\bs v)}}_{\mcQ_1} \ket{\overline{\mcB(\bs w)}}_{\mcQ_1} &=\overline{V_f}^{\otimes K_0} \ket{\overline{\mcB(\bs u)}}_{\mcQ_1} \ket{\overline{\mcB(\bs v)}}_{\mcQ_1} \ket{\overline{\mcB(\bs w)}}_{\mcQ_1} \\
        &= (-1)^{\sum_{i=1}^{K_0} f(u_i v_i w_i)} \ket{\overline{\mcB(\bs u)}}_{\mcQ_1} \ket{\overline{\mcB(\bs v)}}_{\mcQ_1} \ket{\overline{\mcB(\bs w)}}_{\mcQ_1}.
    \end{align}
    
    \paragraph{Step 2:} Next, for each mentioned group of $m$ logical qubits of $\mcQ_1$, we hardcode a subset of $m-s$ logical qubits to $\ket{\overline{0}}$, yielding a new CSS code $\mcQ_2$. The goal of this step is to handle the `logical end' of the qudit-to-qubit code transformation.
    Specifically, consider the $2^s$ field elements in the image of the RMFE, $\Im(\phi_\rmfe)$. Because $\phi_\rmfe$ is $\mathbb{F}_2$-linear, under the basis $\mcB$ of~\Cref{lem:qubit-selfdual}, these elements in $\Im(\phi_\rmfe)$, when viewed as vectors in $\mathbb{F}_2^m$, form an $s$-dimensional linear subspace in the computational basis of $\mathbb{F}_2^m$. For each $i \in K_0$, we take appropriate combinations of $\{\overline{Z}_{i,j}\}_{i \in [m]}$ (possibly including $\pm 1$ signs) to form a maximal group stabilizing the logical computational basis states corresponding to the linear subspace $\mathbb{F}_2^m \backslash \mcB(\Im(\phi_\rmfe))$ within that block $i$ of logical qubits. Then, adding those Pauli Z operators into the stabilizer group of $\mcQ_1$ yields a CSS code $\mcQ_2$ with parameters $[[N_2=N_0m, K_2 = K_0s , D_2 \geq D_0 ]]$. The new code $\mcQ_2$ encodes $K_0$ groups of $s$ logical qubits. 
    Importantly, this step does not change the (block-)transversality of $V_f$ since $V_f$ is diagonal and thus commutes with the added $Z$ stabilizers.
    
    We define a logical computational basis of $\mcQ_2$ as follows. For $\bs v \in (\mathbb{F}_2^s)^{K_0}$, we refer to the entries of $\bs v$ by $v_{i,j}$ for $i \in [K_0]$ and $j \in [s]$. We denote by $v_{i,\cdot}$ the vector $(v_{i,1},\hdots,v_{i,s}) \in \mathbb{F}_2^s$. And $\phi_\rmfe(\bs v)  \triangleq (\phi_\rmfe(v_{1,\cdot}),\hdots,\phi_\rmfe(v_{K_0,\cdot})) \in \mathbb{F}_{2^m}^{K_0}$. The logical computational basis states of $\mcQ_2$ are, for $\bs v \in (\mathbb{F}_2^s)^{K_0}$,
    \begin{align}
        \ket{\overline{\bs v}}_{\mcQ_2} \triangleq \ket{\overline{\mcB(\phi_\rmfe(\bs v))}}_{\mcQ_1}.
    \end{align}
    Then, the action of $V_f^{\otimes N_0}$ is
    \begin{align}
        V_f^{\otimes N_0} \ket{\overline{\bs u}}_{\mcQ_2} \ket{\overline{\bs v}}_{\mcQ_2} \ket{\overline{\bs w}}_{\mcQ_2}= (-1)^{\sum_{i=1}^{K_0} f(\phi_\rmfe(u_{i,\cdot}) \phi_\rmfe(v_{i,\cdot} )
        \phi_\rmfe(w_{i,\cdot}))} \ket{\overline{\bs u}}_{\mcQ_2} \ket{\overline{\bs v}}_{\mcQ_2} \ket{\overline{\bs w}}_{\mcQ_2}.
    \end{align}

We now specify $f$ to be $f(a) \triangleq \sum_{l=1}^s (\psi_\rmfe(a))_l \mod 2$ for $a \in \mathbb{F}_{2^m}$, i.e., the parity of $\psi_\rmfe(a) \in \mathbb{F}_2^s$. Since $\psi_\rmfe$ is $\mathbb{F}_2$-linear, $f$ is $\mathbb{F}_2$-linear, indeed satisfying the transversality requirement of~\Cref{lem:ccz}. Using the RMFE property, we then have $f(\phi_\rmfe(u_{i,\cdot}) \phi_\rmfe(v_{i,\cdot} )
        \phi_\rmfe(w_{i,\cdot}))= \sum_{j=1}^s u_{i,j} v_{i,j} w_{i,j}$. Therefore,
\begin{align}
    V_f^{\otimes N_0} \ket{\overline{\bs u}}_{\mcQ_2} \ket{\overline{\bs v}}_{\mcQ_2} \ket{\overline{\bs w}}_{\mcQ_2}
    &=  (-1)^{\sum_{i=1}^{K_0} \sum_{j=1}^{s} u_{i,j} v_{i,j} 
        w_{i,j}} \ket{\overline{\bs u}}_{\mcQ_2} \ket{\overline{\bs v}}_{\mcQ_2} \ket{\overline{\bs w}}_{\mcQ_2} \\
    &= \overline{CCZ}^{\otimes K_0s} \ket{\overline{\bs u}}_{\mcQ_2} \ket{\overline{\bs v}}_{\mcQ_2} \ket{\overline{\bs w}}_{\mcQ_2}.
    \label{eq:block-transversal}
\end{align}

\paragraph{Step 3:} In this final step, we use MFE to handle the `physical end' of the qudit-to-qubit code transformation. So far we have constructed a qubit code $\mcQ_2$ on which $\overline{CCZ}^{\otimes K_2}$ is realized by the (block-) transversal $V_f^{\otimes N_0}$. We now convert $\mcQ_2$ into a code $\mcQ_3$ to achieve full transversality. Our goal is to concatenate $\mcQ_3$ with a classical code such that the following $3m$-bit (physical) operation
\begin{align}
    V_f: \ket{\mcB(x)} \ket{\mcB(y)} \ket{\mcB(z)} \mapsto (-1)^{f(xyz)} \ket{\mcB(x)} \ket{\mcB(y)} \ket{\mcB(z)}, \qquad \forall x,y,z \in \mathbb{F}_{2^m}
    \label{eq:Vf}
\end{align}
can be emulated by transversal CCZ gates. This is done via the MFE from~\Cref{lem:mfe}. In particular, the embedding $\sigma_\mfe: \mathbb{F}_{2^m} \mapsto \mathbb{F}_2^r$ is an $\mathbb{F}_2$-linear map in a polynomial basis of $\mathbb{F}_{2^m}$ over $\mathbb{F}_2^m$, i.e., the basis in the proof of~\Cref{lem:mfe} has the form $(1,\alpha,\hdots,\alpha^{m-1})$, where $\alpha$ is a primitive element. This basis could be different from the self-dual basis in~\Cref{lem:qubit-selfdual}. However, any two linearly independent bases are related via an invertible matrix $A \in \mathbb{F}_2^{m \times m}$. So the map $\sigma'_\mfe: \mathbb{F}_2^m \mapsto \mathbb{F}_2^r$ defined by $\sigma'_\mfe(\mcB(x)) = \sigma_\mfe (x)$, for $x \in \mathbb{F}_{2^m}$, is linear and injective. In other words, $\Im(\sigma'_\mfe)$ is nothing but a binary linear code which we denote as $C_\mfe$. By definition, it holds for all $x,y,z \in \mathbb{F}_{2^m}$ that
\begin{align}
    \psi_\mfe(\sigma'_\mfe(\mcB(x)) \star \pi_2(\sigma'_\mfe(\mcB(y))) \star \pi_3(\sigma'_\mfe(\mcB(z)))) = xyz,
\end{align}
where $\pi_2, \pi_3$ are the permutations from the definition of MFE.

Thus,
\begin{align}
    f(xyz) &= f(\psi_\mfe(\sigma'_\mfe(\mcB(x)) \star \pi_2(\sigma'_\mfe(\mcB(y))) \star \pi_3(\sigma'_\mfe(\mcB(z)))))\\
    &= \sum_{l=1}^s (\psi_\rmfe(   \psi_\mfe(\sigma'_\mfe(\mcB(x)) \star \pi_2(\sigma'_\mfe(\mcB(y))) \star \pi_3(\sigma'_\mfe(\mcB(z))))) )_l \mod 2\\
    &= \sum_{l=1}^s (M  (\sigma'_\mfe(\mcB(x)) \star \pi_2(\sigma'_\mfe(\mcB(y))) \star \pi_3(\sigma'_\mfe(\mcB(z)))))_l \mod 2\\
    &= \sum_{i=1}^{r} \left( \sum_{l=1}^s M_{l,i} \mod 2 \right) (\sigma'_\mfe(\mcB(x)))_i (\sigma'_\mfe(\mcB(y)))_{\pi_2^{-1}(i)} (\sigma'_\mfe(\mcB(z)))_{\pi_3^{-1}(i)} \mod 2,
\end{align}
where $M \triangleq \psi_\rmfe \circ \psi_\mfe \in \mathbb{F}_2^{s \times r}$ is a linear map because $\psi_\rmfe$ and $\psi_\mfe$ are $\mathbb{F}_2$-linear.

In other words, the quantum operation $V_f$ in~\Cref{eq:Vf} can be implemented transversally by CCZ gates once we have encoded each $m$-bit register into the $[r,m,\cdot]$ code $C_\mfe$ (the distance is unspecified since it's not too important for us), which can be viewed as quantum code $Q_\mfe$ of parameters $[[r,m,1]]$. Define $P \in \mathbb{F}_2^r$ where $P_i = \sum_{l=1}^s M_{l,i} \mod 2$, then
\begin{align}
     CCZ(P) \ket{\overline{\mcB(x)}}_{Q_\mfe} \ket{\overline{\mcB(y)}}_{Q_\mfe} \ket{\overline{\mcB(z)}}_{Q_\mfe} = (-1)^{f(xyz)} \ket{\overline{\mcB(x)}}_{Q_\mfe} \ket{\overline{\mcB(y)}}_{Q_\mfe} \ket{\overline{\mcB(z)}}_{Q_\mfe},
     \label{eq:mfe-transversal}
\end{align}
where $CCZ(P)$ applies CCZ on the triplets of the qubits $i,\pi_2^{-1}(i),$ and $\pi_3^{-1}(i)$, each from a register, whenever $P_i=1$, and $\ket{\overline{\mcB(x)}}_{Q_\mfe} \triangleq \mcE_{Q_\mfe} \ket{\mcB(x)}$ with $\mcE_{Q_\mfe}$ being the encoding isometry of $Q_\mfe$.

Concatenating $\mcQ_2$ with $Q_\mfe$ we obtain the desired code $\mcQ_3 = \mcQ_2 \circ Q_\mfe$, which is still a CSS code. It has parameters $[[N_3= N_0r, K_3=K_0 s, D_3 \geq D_0]]$. Its logical computational basis states are, for $\bs v \in (\mathbb{F}_2^s)^{K_0}$,
\begin{align}
        \ket{\overline{\bs v}}_{\mcQ_3} \triangleq (\mcE_{Q_\mfe})^{\otimes N_0}\ket{\overline{\bs v}}_{\mcQ_2},
\end{align}
Finally, combining~\Cref{eq:block-transversal} and~\Cref{eq:mfe-transversal} we can verify that
\begin{align}
    (CCZ(P))^{\otimes N_0} \ket{\overline{\bs u}}_{\mcQ_3} \ket{\overline{\bs v}}_{\mcQ_3} \ket{\overline{\bs w}}_{\mcQ_3}= \overline{CCZ}^{\otimes K_0s} \ket{\overline{\bs u}}_{\mcQ_3} \ket{\overline{\bs v}}_{\mcQ_3} \ket{\overline{\bs w}}_{\mcQ_3},
\end{align}
concluding the proof.
\end{proof}

\section*{Acknowledgements}
We thank Anurag Anshu, Fernando Brandao, Niko Breuckmann, Thiago Bergamaschi, Alex Dalzell, Jonas Haferkamp, Joe Iverson, Rachel Zhang for useful discussions and pointer to useful references. We especially thank Chris Pattison for extensive discussions regarding codes with multiplication property, from which this work started. We also thank the authors of~\cite{wills2024constantoverheadmagicstatedistillation} for suggesting us to include more details on how constant-space-overhead MSD follows from our codes, which we inattentively omitted in our version 1. This work was done in part while the author was visiting the Simons Institute for the Theory of Computing, supported by NSF QLCI Grant No. 2016245, and finalized while he was an intern at the AWS Center for Quantum Computing. The author is grateful for the hospitality of both places. The author acknowledges support from the NSF Award
No. 2238836 via Anurag Anshu and support from the Harvard Quantum Initiative.

\bibliographystyle{unsrt}
\bibliography{bibliography.bib}

\begin{thebibliography}{10}

\bibitem{bravyi2005universal}
Sergey Bravyi and Alexei Kitaev.
\newblock Universal quantum computation with ideal clifford gates and noisy ancillas.
\newblock {\em Physical Review A—Atomic, Molecular, and Optical Physics}, 71(2):022316, 2005.

\bibitem{zeng2007transversality}
Bei Zeng, Andrew Cross, and Isaac~L Chuang.
\newblock Transversality versus universality for additive quantum codes.
\newblock {\em arXiv preprint arXiv:0706.1382}, 2007.

\bibitem{bravyi2012magic}
Sergey Bravyi and Jeongwan Haah.
\newblock Magic-state distillation with low overhead.
\newblock {\em Physical Review A—Atomic, Molecular, and Optical Physics}, 86(5):052329, 2012.

\bibitem{campbell2012magic}
Earl~T Campbell, Hussain Anwar, and Dan~E Browne.
\newblock Magic-state distillation in all prime dimensions using quantum reed-muller codes.
\newblock {\em Physical Review X}, 2(4):041021, 2012.

\bibitem{campbell2014enhanced}
Earl~T Campbell.
\newblock Enhanced fault-tolerant quantum computing in d-level systems.
\newblock {\em Physical review letters}, 113(23):230501, 2014.

\bibitem{haah2018towers}
Jeongwan Haah.
\newblock Towers of generalized divisible quantum codes.
\newblock {\em Physical Review A}, 97(4):042327, 2018.

\bibitem{vasmer2022morphing}
Michael Vasmer and Aleksander Kubica.
\newblock Morphing quantum codes.
\newblock {\em PRX Quantum}, 3(3):030319, 2022.

\bibitem{haah2017magic}
Jeongwan Haah, Matthew~B Hastings, David Poulin, and D~Wecker.
\newblock Magic state distillation with low space overhead and optimal asymptotic input count.
\newblock {\em Quantum}, 1:31, 2017.

\bibitem{hastings2018distillation}
Matthew~B Hastings and Jeongwan Haah.
\newblock Distillation with sublogarithmic overhead.
\newblock {\em Physical review letters}, 120(5):050504, 2018.

\bibitem{haah2018codes}
Jeongwan Haah and Matthew~B Hastings.
\newblock Codes and protocols for distilling $ t $, controlled-$ s $, and toffoli gates.
\newblock {\em Quantum}, 2:71, 2018.

\bibitem{anwar2012qutrit}
Hussain Anwar, Earl~T Campbell, and Dan~E Browne.
\newblock Qutrit magic state distillation.
\newblock {\em New Journal of Physics}, 14(6):063006, 2012.

\bibitem{vuillot2022quantum}
Christophe Vuillot and Nikolas~P Breuckmann.
\newblock Quantum pin codes.
\newblock {\em IEEE Transactions on Information Theory}, 68(9):5955--5974, 2022.

\bibitem{zhu2023non}
Guanyu Zhu, Shehryar Sikander, Elia Portnoy, Andrew~W Cross, and Benjamin~J Brown.
\newblock Non-clifford and parallelizable fault-tolerant logical gates on constant and almost-constant rate homological quantum ldpc codes via higher symmetries.
\newblock {\em arXiv preprint arXiv:2310.16982}, 2023.

\bibitem{krishna2019towards}
Anirudh Krishna and Jean-Pierre Tillich.
\newblock Towards low overhead magic state distillation.
\newblock {\em Physical review letters}, 123(7):070507, 2019.

\bibitem{goppa1981codes}
Valerii~Denisovich Goppa.
\newblock Codes on algebraic curves.
\newblock In {\em Sov. Math.-Dokl}, volume~24, pages 170--172, 1981.

\bibitem{couvreur2021algebraic}
Alain Couvreur and Hugues Randriambololona.
\newblock Algebraic geometry codes and some applications.
\newblock In {\em Concise encyclopedia of coding theory}, pages 307--362. Chapman and Hall/CRC, 2021.

\bibitem{chen2001some}
Hao Chen.
\newblock Some good quantum error-correcting codes from algebraic-geometric codes.
\newblock {\em IEEE Transactions on Information Theory}, 47(5):2059--2061, 2001.

\bibitem{chen2001asymptotically}
Hao Chen, San Ling, and Chaoping Xing.
\newblock Asymptotically good quantum codes exceeding the ashikhmin-litsyn-tsfasman bound.
\newblock {\em IEEE Transactions on Information Theory}, 47(5):2055--2058, 2001.

\bibitem{bartoli2021certain}
Daniele Bartoli, Maria Montanucci, and Giovanni Zini.
\newblock On certain self-orthogonal ag codes with applications to quantum error-correcting codes.
\newblock {\em Designs, Codes and Cryptography}, 89:1221--1239, 2021.

\bibitem{hernando2020quantum}
Fernando Hernando, Gary McGuire, Francisco Monserrat, and Julio~Jos{\'e} Moyano-Fern{\'a}ndez.
\newblock Quantum codes from a new construction of self-orthogonal algebraic geometry codes.
\newblock {\em Quantum Information Processing}, 19:1--25, 2020.

\bibitem{li2009family}
Zhuo Li, Lijuan Xing, and Xinmei Wang.
\newblock A family of asymptotically good quantum codes based on code concatenation.
\newblock {\em IEEE transactions on information theory}, 55(8):3821--3824, 2009.

\bibitem{la2017good}
Giuliano~G La~Guardia and Francisco Revson~F Pereira.
\newblock Good and asymptotically good quantum codes derived from algebraic geometry.
\newblock {\em Quantum Information Processing}, 16:1--12, 2017.

\bibitem{pereira2021entanglement}
Francisco Revson~F Pereira, Ruud Pellikaan, Giuliano~Gadioli La~Guardia, and Francisco~Marcos De~Assis.
\newblock Entanglement-assisted quantum codes from algebraic geometry codes.
\newblock {\em IEEE Transactions on Information Theory}, 67(11):7110--7120, 2021.

\bibitem{ashikhmin2001asymptotically}
Alexei Ashikhmin, Simon Litsyn, and Michael~A Tsfasman.
\newblock Asymptotically good quantum codes.
\newblock {\em Physical Review A}, 63(3):032311, 2001.

\bibitem{stichtenoth2009algebraic}
Henning Stichtenoth.
\newblock {\em Algebraic function fields and codes}, volume 254.
\newblock Springer Science \& Business Media, 2009.

\bibitem{cascudo2009asymptotically}
Ignacio Cascudo, Hao Chen, Ronald Cramer, and Chaoping Xing.
\newblock Asymptotically good ideal linear secret sharing with strong multiplication over any fixed finite field.
\newblock In {\em Annual International Cryptology Conference}, pages 466--486. Springer, 2009.

\bibitem{cascudo2018amortized}
Ignacio Cascudo, Ronald Cramer, Chaoping Xing, and Chen Yuan.
\newblock Amortized complexity of information-theoretically secure mpc revisited.
\newblock In {\em Advances in Cryptology--CRYPTO 2018: 38th Annual International Cryptology Conference, Santa Barbara, CA, USA, August 19--23, 2018, Proceedings, Part III 38}, pages 395--426. Springer, 2018.

\bibitem{aharonov1999fault}
D~Aharonov and M~{Ben-Or}.
\newblock Fault-tolerant quantum computation with constant error rate.
\newblock {\em arXiv preprint quant-ph/9906129}, 1999.

\bibitem{wills2024constantoverheadmagicstatedistillation}
Adam Wills, Min-Hsiu Hsieh, and Hayata Yamasaki.
\newblock Constant-overhead magic state distillation, 2024.

\bibitem{golowich2024asymptoticallygoodquantumcodes}
Louis Golowich and Venkatesan Guruswami.
\newblock Asymptotically good quantum codes with transversal non-clifford gates, 2024.

\bibitem{nguyen2024}
Quynh~T. Nguyen and Christopher~A. Pattison.
\newblock Quantum fault tolerance with constant-space and logarithmic-time overheads.
\newblock {\em To appear}, 2024.

\bibitem{cross2008}
Andrew~William Cross.
\newblock {\em Fault-tolerant quantum computer architectures using hierarchies of quantum error-correcting codes}.
\newblock Phd thesis, Massachusetts Institute of Technology, 2008.

\bibitem{gottesman2024}
Daniel Gottesman.
\newblock {\em Surviving as a Quantum Computer in a Classical World}.
\newblock 2024.

\bibitem{mullen2013handbook}
Gary~L Mullen and Daniel Panario.
\newblock {\em Handbook of finite fields}, volume~17.
\newblock CRC press Boca Raton, 2013.

\bibitem{tsfasman1982modular}
Michael~A Tsfasman, SG~Vl{\u{a}}dutx, and Th~Zink.
\newblock Modular curves, shimura curves, and goppa codes, better than varshamov-gilbert bound.
\newblock {\em Mathematische Nachrichten}, 109(1):21--28, 1982.

\bibitem{ihara1981some}
Yasutaka Ihara.
\newblock Some remarks on the number of rational points of algebraic curves over finite fields.
\newblock {\em J. Fac. Sci. Tokyo}, 28(3):721--724, 1981.

\bibitem{vladut1983number}
Sergei~G Vladut and Vladimir~Gershonovich Drinfel'd.
\newblock Number of points of an algebraic curve.
\newblock {\em Functional analysis and its applications}, 17(1):53--54, 1983.

\bibitem{garcia1995tower}
Arnaldo Garcia and Henning Stichtenoth.
\newblock A tower of artin-schreier extensions of function fields attaining the drinfeld-vladut bound.
\newblock {\em Inventiones mathematicae}, 121(1):211--222, 1995.

\bibitem{bassa2012towers}
Alp Bassa, Peter Beelen, Arnaldo Garcia, and Henning Stichtenoth.
\newblock Towers of function fields over non-prime finite fields.
\newblock {\em arXiv preprint arXiv:1202.5922}, 2012.

\bibitem{hartshorne2013algebraic}
Robin Hartshorne.
\newblock {\em Algebraic geometry}, volume~52.
\newblock Springer Science \& Business Media, 2013.

\bibitem{tillich2013quantum}
Jean-Pierre Tillich and Gilles Z{\'e}mor.
\newblock Quantum ldpc codes with positive rate and minimum distance proportional to the square root of the blocklength.
\newblock {\em IEEE Transactions on Information Theory}, 60(2):1193--1202, 2013.

\bibitem{breuckmann2021balanced}
Nikolas~P Breuckmann and Jens~N Eberhardt.
\newblock Balanced product quantum codes.
\newblock {\em IEEE Transactions on Information Theory}, 67(10):6653--6674, 2021.

\bibitem{panteleev2022asymptotically}
Pavel Panteleev and Gleb Kalachev.
\newblock Asymptotically good quantum and locally testable classical ldpc codes.
\newblock In {\em Proceedings of the 54th Annual ACM SIGACT Symposium on Theory of Computing}, pages 375--388, 2022.

\bibitem{stichtenoth1988self}
Henning Stichtenoth.
\newblock Self-dual goppa codes.
\newblock {\em Journal of Pure and Applied Algebra}, 55(1-2):199--211, 1988.

\bibitem{driencourt1989criterion}
Yves Driencourt and Henning Stichtenoth.
\newblock A criterion for self-duality of geometric codes.
\newblock {\em Communications in Algebra}, 17(4):885--898, 1989.

\bibitem{stichtenoth2006transitive}
Henning Stichtenoth.
\newblock Transitive and self-dual codes attaining the tsfasman-vladut-zink bound.
\newblock {\em IEEE transactions on information theory}, 52(5):2218--2224, 2006.

\bibitem{grassl1999quantum}
Markus Grassl, Willi Geiselmann, and Thomas Beth.
\newblock Quantum reed—solomon codes.
\newblock In {\em Applied Algebra, Algebraic Algorithms and Error-Correcting Codes: 13th International Symposium, AAECC-13 Honolulu, Hawaii, USA, November 15--19, 1999 Proceedings 13}, pages 231--244. Springer, 1999.

\bibitem{gao1993normal}
Shuhong Gao.
\newblock {\em Normal bases over finite fields.}
\newblock University of Waterloo Waterloo, Canada, 1993.

\bibitem{escudero2023degree}
Daniel Escudero, Cheng Hong, Hongqing Liu, Chaoping Xing, and Chen Yuan.
\newblock Degree-d reverse multiplication-friendly embeddings: constructions and applications.
\newblock In {\em International Conference on the Theory and Application of Cryptology and Information Security}, pages 106--138. Springer, 2023.

\bibitem{xing2007algebraic}
Chaoping Xing and Sze~Ling Yeo.
\newblock Algebraic curves with many points over the binary field.
\newblock {\em Journal of Algebra}, 311(2):775--780, 2007.

\end{thebibliography}

\appendix




\end{document}